\pdfoutput=1
\documentclass[11pt]{article} 
\usepackage{fullpage} 
\usepackage[utf8]{inputenc}
\usepackage[english]{babel} 
\usepackage{amsthm} 
\usepackage{amsmath}
\usepackage{amssymb} 
\usepackage{mathtools,xparse} 
\usepackage{xargs}
\usepackage[pdftex,dvipsnames]{xcolor}
\usepackage[colorinlistoftodos,prependcaption,textsize=tiny]{todonotes}
\usepackage{algorithm} 
\usepackage[noend]{algpseudocode} 
\usepackage{color}
\usepackage[backend=bibtex, style=numeric, maxbibnames=99]{biblatex} 
\usepackage{pgf}
\usepackage{tikz}
\usetikzlibrary{shapes.geometric, arrows,automata} 
\usepackage{tipa} 
\usepackage{mathrsfs}
\usepackage{algorithmicx}
\usepackage{graphicx}
\usepackage[normalem]{ulem}

\algdef{SE}[DOWHILE]{Do}{doWhile}{\algorithmicdo}[1]{\algorithmicwhile\ #1}

\newcommand{\eccconst}{C_e}

\newcommand{\paran}[1]{\left( #1 \right)}

\newcommand{\amdstr}{\eta}

\tikzstyle{startstop} = [rectangle, rounded corners, minimum width=2cm, minimum height=1cm, text centered, draw=black, fill=red!20]
\tikzstyle{dummy} = [rectangle, fill=black]
\tikzstyle{process} = [rectangle, minimum width=2cm, minimum height=1cm, text centered, draw=black]
\tikzstyle{decision} = [diamond, minimum width=2cm, minimum height=1cm, text centered, draw=black]
\tikzstyle{arrow} = [thick,->,>=stealth]

\newcommand{\ecc}[1]{\texttt{ecEnc}\left( #1 \right)}
\newcommand{\ecci}[1]{\texttt{ecDec}\left( #1 \right)}

\newcommand{\silence}[1]{\texttt{IsSilence}\left( {#1} \right)}

\newcommand{\amdenc}[2]{\texttt{amdEnc}\left( #1, #2 \right)}

\newcommand{\amddec}[2]{\texttt{amdDec}\left( #1, #2 \right)}

\newcommand{\iscodeword}[2]{\texttt{IsCodeword}\left( #1, #2 \right)}

\newcommand{\apnr}{\varepsilon}
\newcommand{\errorprob}{\delta}

\newcommand{\jared}[1]{}
\newcommand{\tom}[1]{}
\newcommand{\varsha}[1]{}

\newtheorem{theorem}{Theorem}[section]

\newtheorem{lemma}{Lemma}[section]
\usepackage[utf8]{inputenc}
\usepackage[T1]{fontenc}
\usepackage{xspace}

\newcommand{\encodedMsg}[3]{\mathscr{E}_{#3}\paran{#1,#2}}
\newcommand{\decodeMsg}[2]{\mathscr{D}_{#2}\paran{#1}}
\newcommand{\decodeMsgNoArgs}{\mathscr{D}}

\newcommand{\round}{round\xspace}
\newcommand{\rounds}{rounds\xspace}

\newcommand{\Rounds}{Rounds\xspace}

\newcommand{\failAMD}{\emph{AMD Failure}\xspace}
\newcommand{\failConversionSilence}{\emph{Conversion to Silence}\xspace}
\newcommand{\failKeyInstallation}{\emph{Key Installation}\xspace}

\newcommand{\sendMessage}[0]{\text{SEND-MESSAGE}}
\newcommand{\receiveMessage}[0]{\text{RECEIVE-MESSAGE}}
\newcommand{\terminatedPi}[0]{\emph{terminated in $\pi$}\xspace}

\algnewcommand{\LineComment}[1]{\State \(\triangleright\) #1}
\newcommand{\pfl}{\mathscr{L}}
\newcommand{\hlt}[1]{{\color{red} #1}}

\title{Distributed Computing with Channel Noise}

\author{
	Abhinav Aggarwal\\
	\texttt{abhiag@cs.unm.edu}
	\and
	Varsha Dani\\
	\texttt{varsha@cs.unm.edu}
	\and
	Thomas P. Hayes\\
	\texttt{hayes@cs.unm.edu}
	\and
	Jared Saia\\
	\texttt{saia@cs.unm.edu}
}

\date{}
\begin{document}
\maketitle
\begin{abstract}
A group of $n$ users want to run a distributed protocol $\pi$ over a network where communication occurs via private point-to-point channels. Unfortunately, an adversary, who knows $\pi$, is able to maliciously flip bits on the channels.  Can we efficiently simulate $\pi$ in the presence of such an adversary?

We show that this is possible, even when $L$, the number of bits sent in $\pi$, and $T$, the number of bits flipped by the adversary are not known in advance.  In particular, we show how to create a robust version of  $\pi$ that 1) fails with probability at most $\errorprob$, for any $\errorprob>0$; and 2) sends $\tilde{O}(L + T)$ bits, where the $\tilde{O}$ notation hides a $\log (nL/ \errorprob)$ term multiplying $L$.
 
Additionally, we show how to improve this result when the average
message size $\alpha$ is not constant.  In particular, we give an
algorithm that sends $O( L (1 + (1/\alpha) \log (n L/\errorprob) + T)$
bits.  This algorithm is adaptive in that it does not require \emph{a priori}
knowledge of $\alpha$.  We note that if $\alpha$ is $\Omega\paran{\log (n L/\errorprob)}$, 
then this improved algorithm sends only $O(L+T)$ bits, and is therefore within
a constant factor of optimal.
\end{abstract}

\section{Introduction}

Suppose we start with an $n$-party distributed computational protocol, to be 
executed over a network with point-to-point communication between adjacent parties.
Now, imagine a nearly omniscient adversary, who can flip a subset of the bits sent along 
these communication channels.  This adversary knows the protocol, as well as the players'
inputs, and is unrestricted in terms of which bits she will flip.  Can the players even
hope to carry out their protocol successfully under conditions such as these?

Surprisingly, the answer is yes, in a surprisingly general sense!  However, we will
require some additional conditions.  Firstly, the adversary cannot have an unlimited
budget, or she could effectively cut off all communication forever by tossing a fair coin to decide
whether or not to flip each bit.  We thus restrict the adversary to flipping a total of $T$ bits, where $T$ is a finite number, which is unknown to our algorithm.

Secondly, we require that the channels are private and cannot be read by the adversary, and that players have access to private sources of randomness, which are also not readable by the adversary. \jared{cut: although the adversary can flip arbitrary subsets of bits on the channels - not sure how to work this in}  This is a crucial assumption, 
but one for which there are strong justifications.  If the players have access to shared secret
randomness, i.e. a \emph{one-time pad}, then the channels effectively become private, even if
the adversary can see the traffic on them.  More generally, if the players have access
to \emph{strong cryptography}, then the same conclusion holds in practice.  Also, we may be
interested in robusteness against an \emph{oblivious adversary}, such as a bursty source
of line noise.  Such a source may not actually care about our protocol, but may still time
its bursts in a way that happens to be bad for it.  Finally, our assumption of private channels
seems to be necessary, since without it, the adversary 
can institute a ``man-in-the-middle'' attack against any desired subset of the players.  Since
there is no \emph{a priori} upper or lower bound on the number of bits, $T$, flipped by the adversary, 
such an attack can always fool the players into thinking they have successully executed the protocol
and $T=0$, when in fact, all communication between the two subsets of the players was substituted
by whatever the adversary wanted.

\jared{cut: Secondly, we require that the players have access to private sources of randomness, which are not readable by the adversary.  This is necessary because, without randomness, the adversary would know in advance every message to be sent in the protocol, and, if $T$ is unknown and arbitrary, the adversary could then substitute all messages sent with any messages that she desires. \jared{changed the previous sentence to make it clear that unknown T is what allows for this attack.}}


\paragraph{Our Problem and Result}

\par Consider a group of $n \geq 2$ users who want to run a noisefree 
asynchronous protocol $\pi$, whose length is unknown \emph{a priori}. These 
users are connected via a network of arbitrary topology that is unknown to our algorithm.  The edges of the network represent binary symmetric channels.  Communication on these channels is \emph{synchronous}.

An adversary is able to flip some finite number, $T$, of bits on any subset of the channels of the network at any time steps. The adversary chooses $\pi$, $T$, and which bits to flip on the channels. The adversary also knows our algorithm for transforming $\pi$ to $\pi'$. However, similar to the previous work by Dani et. al~\cite{DHMSM:icalp15,icalp15}, we assume that the advsesary neither knows the private random bits of any user, nor the bits sent over the channels, except when it is possible to infer these from knowledge of $\pi$ and our algorithm. \\

\jared{cut: I don't think this paragraph adds any useful information, and it actually seems to muddy the waters. 
\par Since none of the users are aware of this adversary's strategy or how 
many bits he will flip on the channels, their only defense is 
through some randomness that he has no control over. We present an 
algorithm that given an $\errorprob \in (0,1)$, compiles $\pi$ into a robust protocol $\pi'$, which tolerates such 
an adversary with probability $1 - \errorprob$, without incurring a large overhead on part 
of the users. This algorithm can handle any protocol $\pi$ that is event driven i.e.\ runs in an asynchronous environment, in a simple decentralized manner and so far, is the first such algorithm to have tackled the multiparty case in this setting. }

\jared{I tried to emphasize above that the communication in $\pi'$ is synchronous.  This seems critical for the reworded theorem to make sense.}

\par Our results are be summarized in the following theorem. 

\begin{theorem}
  \label{thm:mainUnbounded}
  Let $\pi$ be a (noisefree) asynchronous protocol for $n \ge 2$ users, and let $\errorprob > 0$.  Then, Algorithm~\ref{multipartyAlgo} compiles $\pi$ into a robust protocol $\pi'$ with the following properties.  For any finite set of $T$ (time,channel) pairs, with probability
at least $1 - \errorprob$, when the adversary flips the bits on those channels at the corresponding times, the following hold:
\begin{enumerate}
\item The simulation succeeds.  That is, each user stops after a finite number of time steps, and each ends in possession of his share of a valid transcript of $\pi$. 
\item The total number of bits sent is $\mathcal{O} \paran{L\paran{1 + \frac{1}{\alpha}\log \paran{\frac{n L}{\errorprob}}} + T}$.  Here $L$ is the total number of
bits and $\alpha$ is the average message length, for the particular transcript of $\pi$ that was generated.
\item The total latency of $\pi'$ is $\mathcal{O} \paran{\underset{p}{\max} \left \{ \Lambda_p\paran{1 + \frac{1}{\alpha}\log \paran{\frac{n(L+T)}{\errorprob}} + T_p \right \}}}$, where $p$ ranges over all communication paths in the asynchronous simulation of $\pi$, $\Lambda_p$ is the latency of $p$, and $T_p$ is the total number of bits flipped by the adversary on edges in $p$.
\end{enumerate}
\end{theorem}

We note the following.  First, our algorithm $\pi'$ always sends $\tilde{O}(L + T)$ bits, where the $\tilde{O}$ notation hides a logarithmic term in $n$, $L$ and $\errorprob$.  Second, when $\alpha$ is $\Omega\paran{\log (n L/\errorprob)}$, our algorithm sends only $O(L+T)$ bits, and is thus within a constant factor of optimal.  Finally, we stress that our algorithm requires no \emph{a priori} knowledge of $T$, $L$, $\alpha$, or the network topology, but our algorithm does require knowledge of $n$, or at least a polynomial good estimate of $n$.

\subsection{Related Work}

\paragraph{Interactive Communication} Our work is related to interactive communication.  The problem of interactive communication asks how two parties can run a protocol $\pi$ over a noisy channel.  This problem was first posed by Schulman~\cite{schulman:deterministic,schulman:communication}, who describes a deterministic method for simulating interactive protocols on noisy channels with only a constant-factor increase in the total communication complexity. This initial work spurred vigorous interest in the area (see~\cite{braverman:coding} for an excellent survey).  \\

\par Schulman's scheme tolerates an adversarial noise rate of $1/240$, even if the adversary is \emph{not} oblivious.  It critically depends on the notion of a {\it tree code} for which an exponential-time construction was originally provided. This exponential construction time motivated work on  more efficient constructions~\cite{braverman:towards-deterministic,peczarski:improvement,moore:tree}.  There were also efforts to create alternative codes~\cite{gelles:efficient,ostrovsky:error}.  Recently, elegant computationally-efficient schemes that tolerate a constant adversarial noise rate have been demonstrated~\cite{brakerski:efficient,ghaffari:optimal2}. Additionally, a large number of results have improved the tolerable adversarial noise rate~\cite{brakerski:fast,braverman:towards,ghaffari:optimal,franklin:optimal,braverman:list}, as well as tuning the communication costs to a known, but not necessarily constant, adversarial noise rate~\cite{haeupler2014interactive}. \\
	
		\par 	
	Our paper builds on a result on interactive communication by Dani et al~\cite{icalp15}, which in contrast to previous work, assumes private channels, but tolerates an unknown number of bit flips by the adversary.  Their algorithm sends $L + O \left( \sqrt{L(T+1)\log L} + T \right)$ bits in expectation.  They show that private channels are necessary in order to tolerate  unknown $T$.\\
	
	Recently, Braverman et al~\cite{braverman2015constant} show a
        strong lower bound for coding schemes for multiparty
        interactive communication when $\pi$ is a  synchronous
        protocol running on a star network, even when noise is just
        stochastic. Our result circumvents this lower bound since we
        require that $\pi$ be an asynchronous protocol. \\
	
		

  Recent work by Censor-Hillel, Gelles and
  Haeupler~\cite{censorhillel2017}, like ours, shows how to make an
  arbitrary asynchronous distributed protocol robust against
  adversarial noise.  In their case, the new protocol will succeed if
  up to $\Theta(1/n)$ of the messages are corrupted, and the overhead
  is a factor of $O(n \log^2 n)$ in the total communication.  In this
  setting, the fraction $1/n$ of corrupted messages is optimal,
  because the adversary can simply cut off all communication to a
  particular player.  Notably, their algorithm proceeds by essentially
  reducing to the case where the communication network is a tree. \\

Our work is not directly comparable to that of
\cite{censorhillel2017}, because we do not quantify our results
in terms of the fraction of messages that are corrupted, and more
importantly, we assume that communication takes place 
over private channels.  In our case, the stronger assumption of
private channels means that up to a $1/\log(nL)$ fraction of 
message bits can be corrupted, and our algorithm will still 
succeed with high probability, with a cost overhead that is only
$\log(nL)$.  In our setting, the adversarial strategy of trying to cut
off all communications to and from a single player requires corruption
of much more than a $1/n$ fraction of the bits, because our protocol
can detect the noise, and consequently increases the fraction of the 
total communication involving the beleaguered player.

\paragraph{Rateless Codes}  Rateless error correcting codes enable generation of potentially an infinite number of encoding symbols from a given set of source symbols with the property that given any subset of a sufficient number of encoding symbols, the original source symbols can  be recovered. Fountain codes~\cite{mackay2005fountain} and LT codes~\cite{palanki2004rateless,luby2002lt,hashemi2014near} are two classic examples of rateless codes. 	Erasure codes employ feedback for stopping transmission~\cite{palanki2004rateless,luby2002lt} and for error detection~\cite{hashemi2014near} at the receiver. Critically, the feedback channel is typically assumed to be noise free. We differ from this model in that we allow the adversary to flip bits on the feedback channel. Additionally, we tolerate bit flips, while most rateless codes tolerate only bit erasures.

	
\subsection{Formal Model}

\paragraph{Asynchrony of $\pi$}
Recall that a synchronous network is one in which there is a global clock and each message sent in the network takes one time step to be transmitted. On the other hand, an asynchronous network is one in which messages may be arbitrarily and adversarially delayed. In particular messages need not be received in the same order as they are sent. However, every message must be eventually delivered in a finite amount of time. \\


\par As an example, consider a synchronous leader election algorithm in which the 
users start by waiting for their ID number of clock ticks. At the end of this, 
if a user finds no other proclamations of a leader, he proclaims himself as the 
leader. This protocol cannot be simulated in our model. To see this, suppose processor 19 is trying to decide if he is the leader.  If he hears meaningless noise on the channel for multiple time steps of $\pi'$, then he does not know whether the adversary is simply inserting this noise on a silent channel, or whether there is a processor 18 whose leadership proclamation is getting corrupted. Even if  he does not hear noise, i.e. if he hears silence all throughout, then he knows that nobody has transmitted anything to him, but he does not know how many time-steps of $\pi$ have passed in the view of the putative processor 18.

\paragraph{Synchrony of $\pi'$} We assume that our protocol $\pi'$ runs in a \emph{synchronous} network.  In particular, $\pi'$ critically relies on the fact that all processors have synchronized clocks and that communication time is fixed and known. We define a \emph{time step} as the amount of time that it takes to send one bit over a channel, assuming it to be the same for all the channels.  As is standard in distributed computing, we assume that all local computation is instantaneous. 

\paragraph{Silence on the channel} Similar to~\cite{icalp15}, when neither user sends on a channel, we say that the channel is silent. In any contiguous sequence of silent channel steps, the bit received on the channel in the first step is set by the adversary for free. By default, the bit received in the subsequent steps of the sequence remains the same, unless the adversary pays for one bit flip each time it wants to change the value of the bit received in any contiguous sequence of silent steps.

\paragraph{Additional Assumptions} We assume that all users know the desired error tolerance $\errorprob$ for $\pi'$, as well as the number of users in the network, $n$.  However, the users do not know the values of $L$ and $T$, and they have know knowledge of the underlying topology of the network connecting the users. We further assume that $\pi$ has a mechanism for users to detect the start and end of messages they receive. We emphasize that this mechanism need not necessarily involve using end-of-word symbols. Techniques such as prefix-free codes that do not use any special symbols provide a good alternative. 

\subsection{Overview of Our Result}
In simulation of $\pi$ using a synchronous protocol $\pi'$, we assume a global clock which \emph{ticks} every time step, and allows all the users to determine how long messages they should send at any time step $t \geq 1$, which is only a function of $n, \errorprob$ and $t$. Once all messages of $\pi$ have been exchanged in $\pi'$, we are sure that the users have terminated and the protocol has ended. Our approach makes this possible without the use of any synchronization messages across the network, yet keeping all communication as local to the users as possible. \\

\par A key technical challenge here is termination. The length of the simulated protocol is unknown, so the users will likely not terminate at the same time. After a user has terminated, it is a challenge for its neighbors to detect this fact based on the bits received over the noisy channel. Our approach makes use of authentication keys, which every user generates randomly for himself and shares with whoever wants to communicate with him. We use this key in a way that it makes unlikely for the adversary to successfully forge messages or silence on the channels, without detection by the affected users. Hence, a sequence of silent channel steps will suffice, with desired probability, to indicate the users of the termination of their neighbors. \\

\par Our main result is stated in Theorem~\ref{thm:mainUnbounded}. The proof of this theorem and details of how we deal with the challenges above are presented in the subsequent sections. 

\subsection{Rate}
\varsha{This needs to be put in somewhere}  \tom{How about here?  The math looks ok to me in this section, but it needs to be wrapped up at the end,
perhaps with some kind of corollary, or at least a more clear statement of the significance of the result.}
In most prior work it is assumed that the noise rate of the channel(s) is known and may be used as a parameter in the design of the algorithm. The parameter of interest then is the rate of the designed code. In our work we do not assume that the noise rate is known in advance, but only require that the adversary flip a finite number of bits. Nevertheless, in order to compare our result with other work, we compute the coding rate of or algorithm  as a function of the \emph{a posteriori} noise rate. Such a comparison is only meaningful when the adversry's total budget is less than $L$, the length of $\pi$, so for this section we will assume that $T<L$.

Let $L'$ denote the length of $\pi'$. Theorem~\ref{thm:mainUnbounded} states that our algorithm achieves 
\[
L' \le C L \left(1+ \frac1{\alpha}\log\left(\frac{n L}{\errorprob}\right)  \right) + CT 
\]
for some constant $C$, where $\errorprob$ is the permissible failure probability for the algorithm and $\alpha$ is the average message length in $\pi$. Furthermore, since $T<L$ this translates to the absolute upper bound 
\begin{align}
L' &\le C L \left(2+ \frac1{\alpha}\log\left(\frac{n(L+T)}{\errorprob}\right)\right) \label{eqn:ubL'} \\
&\le C L \left(2+ \frac1{\alpha}\log\left(\frac{2nL}{\errorprob}\right)\right) \label{eqn:absubL'}
\end{align}

Let $\apnr = T/L'$ be the \emph{a posteriori} noise rate.
Then $T= \apnr L'$. Making this substitution for $T$ in \eqref{eqn:ubL'}, and using \eqref{eqn:absubL'} to bound the $L'$ inside the log, we have 
\begin{align*}
L' &\le C L \left(2+ \frac1{\alpha}\log\left(\frac{n\left(L+\apnr C L \left(2+ \frac1{\alpha}\log\left(\frac{2nL}{\errorprob}\right)\right)\right)}{\errorprob}\right)\right)
\end{align*}
Dividing by $L$,
\begin{align*}
\frac{L'}{L} &\le 2C+ \frac{C}{\alpha}\log\left(\frac{n\left(L+\apnr C L \left(2+ \frac1{\alpha}\log\left(\frac{2nL}{\errorprob}\right)\right)\right)}{\errorprob}\right)  \\
&= 2C+ \frac{C}{\alpha}\left( \log(nL/\errorprob) +\log\left(1+\apnr \left(2C+\frac{C\log(2nL/\errorprob)}{\alpha}\right)\right)\right)\\
&\le 2C+ \frac{C\log(nL/\errorprob)}{\alpha} +\frac{\apnr C}{\alpha} \left(2C+\frac{C\log(2nL/\errorprob)}{\alpha}\right) \\
&\le \left(2C + \frac{C\log(2nL/\errorprob) }{\alpha}\right) \left(1 +\frac{\apnr C}{\alpha} \right)
\end{align*}

To summarize, for the worst case where $\alpha=1$, the above shows that we achieve a coding rate that increases linearly with the noise rate $\apnr$ and with the logarithm of $nL/\errorprob$. In particular, the coding rate is $O((1 + \apnr) \log (nL/\errorprob))$.  For arbitrary $\alpha$, we achieve a coding rate of $O((1 + \apnr/\alpha) \log (nL/\errorprob)/\alpha)$. \jared{I'm kind of worried that we get the $\apnr/\alpha$ term here.  Is this really right???}

\subsection{Paper Organization}
The rest of this paper is organized as follows. In section~\ref{sec:unboundedTalgo}, we describe our main algorithm for the case where all messages in $\pi$ are exactly $1$ bit in length. We prove  this algorithm is correct and analyze resource costs in Section~\ref{sec:unboundedTanal}. In Sections~\ref{sec:alpha} and~\ref{sec:unboundedTanalArbLength}, we describe and analyze our algorithms for the case where messages in $\pi$ are of an arbitrary size average (thereby achieving better results for some values of the average size $\alpha$).  Finally, we conclude and give directions for future work in Section~\ref{sec:conclusion}.

\section{Our Algorithm}
\label{sec:unboundedTalgo}

We describe our algorithm in this section.  For now, we assume that 
$\pi$ consists of one-bit messages.
\subsection{Notation and Definitions}
Some helper functions and notation used in our algorithm are described here. 

\paragraph{Algebraic Manipulation Detection Codes} Our algorithm makes critical use of Algebraic Manipulation Detection (AMD) codes from~\cite{amd}. For a given $\amdstr > 0$, called the strength of AMD encoding, these codes provide three functions: $\texttt{amdEnc}$, $\texttt{amdDec}$ and $\texttt{IsCodeword}$. The function $\amdenc{m}{\amdstr}$ creates an AMD encoding of a message $m$. The function $\iscodeword{m}{\amdstr}$ takes a message $m$ and returns true if and only if there exists some message $m'$ such that $\amdenc{m'}{\amdstr} = m$. The function $\amddec{m}{\amdstr}$ takes a message $m$ such that $\iscodeword{m}{\amdstr}$ and returns a message $m'$ such that $\amdenc{m'}{\amdstr} = m$. We summarize the results from~\cite{amd} to highlight the important properties of these functions in the following lemma.

\begin{theorem}
	\label{thm:amd}
	There exist functions $\texttt{amdEnc}$, $\texttt{amdDec}$ and $\texttt{IsCodeword}$, such that for any $\amdstr \in \left( 0,\frac{1}{2} \right]$ and any bit string $m$ of length $x$:
	\begin{enumerate}
		\item $\amdenc{m}{\amdstr}$ is a string of length $x+2 \log \paran{\frac{1}{\amdstr}}$.
		\item $\iscodeword{\amdenc{m}{\amdstr}}{\amdstr}$ and $\amddec{\amdenc{m}{\amdstr}}{\amdstr}=m$.
		\item For any bit string $s\neq 0$ of length $x$, we have $\Pr \left( \iscodeword{\amdenc{m}{\amdstr} \oplus s}{\amdstr} \right) \leq \amdstr$.
	\end{enumerate}
\end{theorem}

\paragraph{Error-correcting Codes} These codes enable us to encode a message so that it can be recovered even if the adversary corrupts a third of the bits. We will denote the encoding and decoding functions by \texttt{ecEnc} and \texttt{ecDec}, respectively. The following theorem, established by the results in~\cite{reedSolomon}, gives the properties of these functions. \\

\varsha{Add citation to proper theorems in the cited papers.}

\begin{theorem}
	\label{thm:ecc}
	~\cite{reedSolomon} There exists a constant $\eccconst > 0$ such that for any message $m$, we have $|\ecc{m}| \leq \eccconst |m|$. Moreover, if $m'$ differs from $\ecc{m}$ in at most one third of its bits, then $\ecci{m'}=m$.
\end{theorem} 

With respect to this constant $\eccconst$, we define two more constants $C_1 = 12\eccconst + 76$ and $C_2 = 32\eccconst + 115$, which will be used in our algorithm. \\

\par Finally, we observe that the linearity of \texttt{ecEnc} and \texttt{ecDec} ensure that when the error correction is composed with the AMD code, the resulting code has the following properties:
\begin{enumerate}
	\item If at most a third of the bits of the message are flipped, then the original message can be uniquely reconstructed by rounding to the nearest codeword in the range of \texttt{ecEnc}.
	\item Even if an arbitrary set of bits is flipped, the probability of the change not being recognized is at most $\amdstr$, i.e.\ the same guarantee as for the plain AMD codes.
\end{enumerate}  
This is because the error-correcting code is linear, so when noise $\eta$ is added by the adversary to the codeword $x$, effectively what happens is 
that the decoding function rounds the noise to the nearest codeword.  Thus $\ecci{x+\eta} = \ecci{x} + \ecci{\eta} = m + \ecci{\eta}$, where $m$ is the AMD-encoded message. But now $\ecci{\eta}$ is an obliviously selected string added to the AMD-encoded codeword, and hence the result is very
unlikely to be a valid message unless $\ecci{\eta} = 0$.  \\

For a string $s$, we use the notation $s[i]$ to denote the $i^{th}$ bit of $s$ and $s[i,j]$ to denote the substring $(s[i], s[i+1], \dots, s[j-1])$.
We let $|s|$ denote the length of string $s$, and use the conventions that $s = (s[0], s[1], \dots, s[|s|-1])$, and for $j > |s|$, $s[i,j] = s[i, |s|]$.

\paragraph{Silence} We define the function $\silence{s}$ to return true iff the string $s$ has fewer than $|s|/3$ bit alternations. We also define $\mathscr{S}_\ell = \{ s \in \{0,1 \}^\ell \mid \silence{s}\}$. We drop the subscript when $\ell$ is clear from the context. \\

\subsection{Algorithm Overview}
\label{sec:overview}
For our algorithm, we assume that each pair of neighboring users communicates over a dedicated channel. The algorithm proceeds in rounds, each of which consists of the following steps. 
\begin{enumerate}
  \item If  $u$ has a message for $v$, he initiates a message exchange by asking $v$ for a key.
  \item Upon receipt of this key, $u$ sends the message along with the key.
  \item $v$ terminates the message exchange upon successful authentication and retrieval of the message.
  \item $u$ terminates the message exchange upon hearing silence from $v$.
\end{enumerate}   
This goes on until all the messages in $\pi$ have been communicated to the intended recipients.

\subsubsection*{\Rounds} For each message $m$ in $\pi$ that needs to be sent from some user $u$ to his neighbor $v$, we communicate $m$ through a sequence of exchanges between users $u$, referred to as \emph{Alice}, and $v$, referred to as \emph{Bob}, in $\pi'$ using Algorithms~\ref{aliceAlgo} and~\ref{bobAlgo}, respectively. As mentioned above, the sequence of time steps corresponding to steps 1-4 of the algorithm overview constitute what we call a \emph{round}. Thus, each round of $\pi$ consists of exactly four \emph{words}, one for each of the steps 1-4. The length of each word in round $r$ is denoted $w_r$. This depends on the round number, $r$ and is therefore a function of the time step, that can be computed independently by each user using the clock. We note that $w_r$ gradually increases with $r$. \\

We first assume that $\pi$ has single bit messages. We will extend our results to arbitrary message lengths in Section~\ref{sec:alpha}.  \\

Now, since each message in $\pi$ is just a single bit, it takes exactly one round to be communicated in $\pi'$, if no successful corruption happens, and more otherwise. If a round for some message $m$ is corrupted, we attempt resending $m$ in the subsequent round, possibly with the security increased due to the potentially increased word length. For technical reasons, Alice needs to distinguish between successive messages from $\pi$.  This is because Alice and Bob may have different views on whether a particular round was successful and there may be times when Alice is resending a message that Bob has already received. If Bob encounters two progressive rounds with no silent round in between, which contain the same message, he needs to distinguish between whether Alice is resending the message or whether Alice's next message happens to be the same bits as the previous one.  To disambiguate these cases, Alice appends a bit $b$ to each message $m$ where $b$ is the parity of index of message $m$. We will denote the pair $(m,b)$ by $M$. \\

For convenience, we assume that there are two bidirectional channels between each pair of (neighboring) users, one for each user to initiate a \round with the other. Note that, by time-slicing, we could achieve the same effect on a single bidirectional channel between each pair of users, at the cost of a factor $2$ increase in the number of time steps. \\

Both Alice and Bob generate their words for round $r$ using a function $\encodedMsg{x}{k}{r}$, described below, which returns the encoding of the word's content $x$ using the key $k$ based on the security settings for round $r$. Here $x$ may be the message $m$ from $\pi$, a special keyword \texttt{KEY?} used by Alice to request Bob's key, or Bob's key for the round.  The key $k$ is a string of length $2 \left \lceil \log \frac{4n \pi i}{\sqrt{\errorprob}} \right \rceil$ bits. In Alice's first call to $\mathscr{E}$, $k$ is a random string. In Bob's call, $k$ is the key he received from Alice in the previous word, but $x$ is a random string of the appropriate length. In Alice's second call, $k$ is the key she received from Bob in the previous word. Alice and Bob generate fresh random keys for each round in which they desire to send a message. \\


The security settings are arranged so that the word length in round $r$ is given by:
The word length in round $r$ grows logarithmically, and is given by the following formula.:
\begin{equation}
\label{eq:wordLength}
	w_r = 300 \left\lceil \log(nr/\errorprob) \right\rceil.     
\end{equation}
\tom{Simplified this based on a quick read over.  Was a much messier thing.  As long as it is big enough, it's ok, since we can always pad
the words to fill up the length. } \\

The function $\encodedMsg{x}{k}{r}$ is formally defined as follows.  First, it encodes the pair $(x,k)$ using an AMD code with $\amdstr_r = \frac{\errorprob}{2n^2 \pi^2 r^2}$.  The result is then encoded with a (1/3)-error-correcting code, such as a Reed-Solomon code.  Finally, enough uniformly random bits are appended to bring the total word length up to $w_r$.  These final random bits are added to ensure that, even if corrupted, any deliberately sent word is very unlikely to be mistaken for silence. \\


Similar to $\encodedMsg{x}{k}{r}$, we define a function $\decodeMsg{m}{r}$ which returns either a pair $(x',k')$ such that $\encodedMsg{x'}{k'}{r} = m$ or returns $(\bot,\bot)$, if no such $(x',k')$ exists. It begins by stripping off the padding at the end of $m$ to obtain a shorter string $m'$. 
Then it decodes the error correction, computing $m'' = \ecci{m'}$.   If $\iscodeword{m''}{\amdstr_r}$, then $\decodeMsgNoArgs$ outputs $\amddec{m''}{\amdstr}$, otherwise it returns $(\bot,\bot)$.  \\

The flow of information and control in a round is illustrated in the flowchart in  Fig.~\ref{fig:statediagram}, and specified in detail in 
Algorithms~\ref{aliceAlgo} and~\ref{bobAlgo}.  \\

At the beginning of each round, Bob must listen for a key request during the first $w_r$ time steps. If no valid request is received, he idles until the start of the next round. Thus Bob is active in every round. For her part, Alice only participates in a round if, in $\pi$, she has a message to send to Bob.  \\

There are certain events which may cause our algorithm to fail. More specifically, it is possible that the adversary (1) converts one AMD codeword to another, so that the decoded content is different from the content intended; (2) converts a non-silence word into silence; and (3) correctly guesses some user's key and uses it to communicate bits that are not in $\pi$. We will discuss these failure events in detail in Section~\ref{sec:catastrophicFailure} and analyze their probabilities of occurrence in Section~\ref{sec:probFailure}. However, for the remainder of this section, we assume that none of these failure events happen. \\

In what follows we will describe various scenarios for what may happen in a round during the execution of the algorithm, for a particular bidirectional channel. There are $n^2$ such channels, and the same round may have different scenarios enacted on it on different channels.
Moreover, the views of Alice and Bob may differ on which scenario was enacted. 

\subsubsection*{Silent Rounds}

Bob listens to the first word on the channel and hears silence. Since the adversary cannot manufacture silence, Alice had no message for him and nothing further happens in the round. In this case both views agree that the round was silent.

\subsubsection*{Progressive Rounds}

These are rounds in which the number of bits flipped by the adversary is small enough that it is handled by the correction schemes. Such rounds proceed as follows.\\

Alice has a message $m$ for Bob. She requests his key using the resend bit and the keyword \texttt{KEY?}. Bob decodes Alice's key request and obtains her key. He generates his own random key and sends it to Alice using his knowledge of her key to authenticate his message as actually coming from him. Alice in turn decodes Bob's communication, obtains his key and 
sends $m$, using Bob's key to authenticate her communication. Bob correctly receives $m$, records it in his transcript of $\pi'$ and remains silent. Alice hears silence on the line and decides that Bob has successfully received her message. Here both Alice and Bob agree that the round was successful. 

\subsubsection*{Corrupted Rounds}

These are rounds when the adversary is active, and corrupts one or more of the words in the round. 

\paragraph{Case 1} Alice is silent, but the adversary sends Bob a key request. Then Bob sends 
Alice his key, but she is not listening and remains silent. At this point the adversary  can say whatever he wants on the channel. However since he does not know Bob's key, he cannot authenticate his message. So Bob receives an invalid communication, and responds with noise, but again Alice is not listening. Thus such a round is corrupted in Bob's view, but silent in Alice's. Note that Bob may realize that the round was silent in Alice's view at a later stage, but we still account for this as a corrupted round, since Bob has already incurred a cost for the corruption. 

\paragraph{Case 2} The adversary corrupts Alice's key request. Then Bob does not receive a valid key request, and hence remains silent the rest of the round. Since Bob is silent, the adversary may say whatever he wants on the channel. In particular he may try to pretend he is Bob and send Alice a bogus key. However since he cannot guess Alice's key, he cannot authenticate himself as Bob. So Alice receives an invalid communication, and stays silent until the end of the round. 

\paragraph{Case 3} Bob receives Alice's key request, but the adversary corrupts his communication containing his key. Since he cannot corrupt one AMD codeword into another, Alice cannot decode the message into a key for Bob, causing her to remain silent for the remainder of the round. Again, the adversary cannot install a bogus message from Alice because he cannot guess Bob's key, so Bob receives an invalid communication.

\paragraph{Case 4} The adversary is inactive for the first half of the round and Alice receives Bob's key. Then the adversary corrupts her communication of $m$. Then Bob receives garbage at his end and injects noise into the channel. Since the adversary cannot convert this noise into silence, Alice knows that the round has failed.

In cases 2, 3 and 4, Alice and Bob both know that the round has been corrupted. Alice will retry sending her message in the next round. 

\paragraph{Case 5} The round succeeds all the way to the point where Bob correctly receives $m$, decides the round is successful and remains silent. Then the adversary injects noise onto the channel causing Alice to think the round has failed. In this case again, Alice and Bob's views differ. However this is not a problem because Alice will simply resend the message, and Bob will receive it again, in the next round that is successful in his view. However Bob can recognize the message as a repetition because of the parity bit. At that stage, of course, Bob realizes that the round was in fact, corrupted, but this makes no difference to his future actions.

\subsection{Failure Events}
\label{sec:catastrophicFailure}
As mentioned before, there are certain events can cause catastrophic failure from which the algorithm cannot recover. 
\begin{enumerate}
	\item \emph{Failure of AMD codes} : The adversary's bit flips  happen to convert an AMD codeword into another valid AMD codeword. In this case the decoded content differs from the intended content resulting in authentification failure, when the content was Bob's key, or incorrect simulation of $\pi$ when the content was Alice's message.
	\item \emph{Conversion to silence} : The adversary's bit flips are such that the resulting word looks like silence to its recipient. If the noise sent by Bob in line~\ref{alg:bobNoise} to request a resend is converted to silence, then Alice incorrectly assumes that Bob has received her message and stops transmitting it. This results in an incorrect transcript of $\pi$. Other words being converted to silence could result in a player being silent on the following word, which in conjuntion with guessing the key (see below) could result in the adversary being able to say whatever he wants on the channel.
	\item \emph{Guessing the key} : On a round when Alice is silent in $\pi$, the adversary can install a bogus key request on her channel to Bob. Ordinarily this is not a problem because when Bob responds with his key, the adversary cannot read it, and therefore cannot send Bob an authenticated message. However, if he happens to guess Bob's key, then he can send Bob a message purporting to be from Alice, resulting in an incorrect simulation of $\pi$.  
\end{enumerate}

We will show in Section~\ref{sec:probFailure} that the probability of such a catastrophic failure over the entire run of the algorithm is at most $\errorprob.$

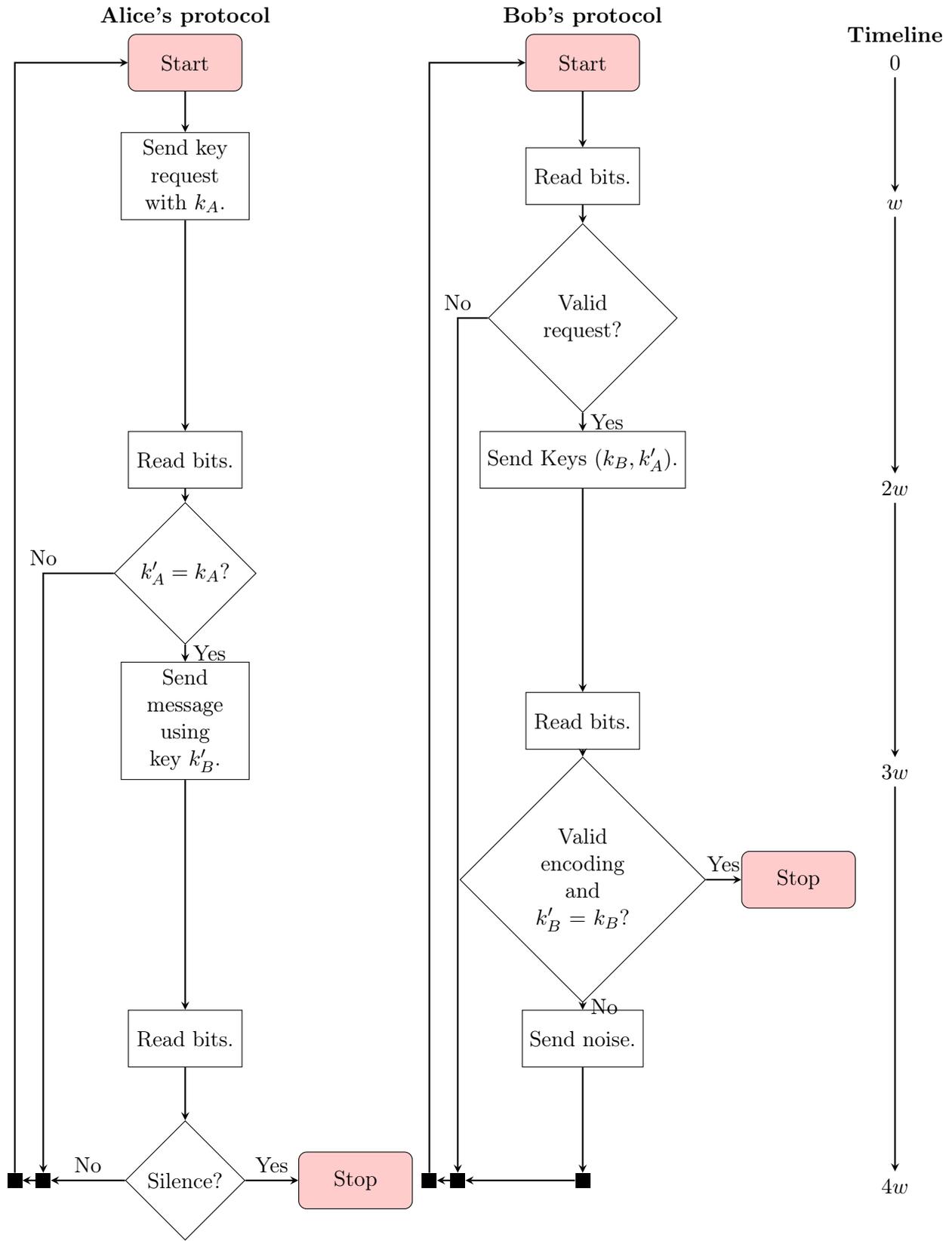
\begin{figure}
\begin{center}
\begin{tikzpicture}[node distance=2cm]

\node (a_start) [startstop] {Start};
\node (b_start) [startstop, right of = a_start, node distance = 7cm] {Start};
\node (aname) [above of = a_start, node distance = 0.8cm] {\textbf{Alice's protocol}};
\node (bname) [above of = b_start, node distance = 0.8cm] {\textbf{Bob's protocol}};
\node (a_requestKey) [process, below of = a_start, text width = 2cm] {Send key request with $k_A$.};
\node (b_readBits) [process, below of = b_start] {Read bits.};
\node (b_validKey) [decision, below of = b_readBits, node distance = 2.5cm, text width = 2cm] {Valid request?};
\node (b_sendKey) [process, below of = b_validKey, node distance = 2.5cm] {Send Keys $(k_B,k_A')$.};
\node (a_readBits) [process, left of = b_sendKey, node distance = 7cm] {Read bits.};
\node (a_validKey) [decision, below of = a_readBits, node distance = 2cm] {$k_A'=k_A$?};
\node (a_sendMessage) [process, below of = a_validKey, node distance = 2.6cm, text width = 2cm] {Send message using key $k_B'$.};
\node (b_readBits2) [process, right of = a_sendMessage, node distance = 7cm] {Read bits.};
\node (b_validMessage) [decision, below of = b_readBits2, node distance = 2.8cm, text width = 2cm] {Valid encoding and $k_B'=k_B$?};
\node (b_sendNoise) [process, below of = b_validMessage, node distance = 2.8cm] {Send noise.};
\node (b_stop) [startstop, right of = b_validMessage, node distance = 3.8cm] {Stop};
\node (a_readBits2) [process, left of = b_sendNoise, node distance = 7cm] {Read bits.};
\node (a_silence) [decision, below of = a_readBits2, node distance = 2.5cm] {Silence?};
\node (a_dummy2) [dummy, left of = a_silence, node distance = 2.5cm] {};
\node (a_dummy3) [dummy, left of = a_dummy2, node distance = 0.5cm] {};
\node (a_stop) [startstop, right of = a_silence, node distance = 3cm] {Stop};
\node (b_dummy3) [dummy, below of = b_sendNoise, node distance = 2.5cm] {};
\node (b_dummy2) [dummy, left of = b_dummy3, node distance = 2.2cm] {};
\node (b_dummy4) [dummy, left of = b_dummy2, node distance = 0.5cm] {};

\draw [arrow] (a_start) -- (a_requestKey);
\draw [arrow] (a_requestKey) -- (a_readBits);
\draw [arrow] (a_readBits) -- (a_validKey);
\draw [arrow] (a_validKey) -| node[anchor=south] {No} (a_dummy2);
\draw [arrow] (a_dummy3) |- (a_start);
\draw [arrow] (a_validKey) -- node[anchor=west] {Yes} (a_sendMessage);
\draw [arrow] (a_sendMessage) -- (a_readBits2);
\draw [arrow] (a_readBits2) -- (a_silence);
\draw [arrow] (a_silence) -- node[anchor=south] {No} (a_dummy2);
\draw [arrow] (a_dummy2) -- (a_dummy3);
\draw [arrow] (a_silence) -- node[anchor=south] {Yes} (a_stop);

\draw [arrow] (b_start) -- (b_readBits);
\draw [arrow] (b_readBits) -- (b_validKey);
\draw [arrow] (b_validKey) -| node[anchor=south] {No} (b_dummy2);
\draw [arrow] (b_dummy4) |- (b_start);
\draw [arrow] (b_validKey) -- node[anchor=west] {Yes} (b_sendKey);
\draw [arrow] (b_sendKey) -- (b_readBits2);
\draw [arrow] (b_readBits2) -- (b_validMessage);
\draw [arrow] (b_validMessage) -- node[anchor=south] {Yes} (b_stop);
\draw [arrow] (b_validMessage) -- node[anchor=west] {No} (b_sendNoise);
\draw [arrow] (b_sendNoise) -- (b_dummy3);
\draw [arrow] (b_dummy3) -- (b_dummy2);
\draw [arrow] (b_dummy2) -- (b_dummy4);

\node (time0) [right of = b_start, node distance = 5.5cm] {$0$};
\node (time1) [below of = time0, node distance = 2.5cm] {$w$};
\node (time2) [below of = time1, node distance = 5cm] {$2w$};
\node (time3) [below of = time2, node distance = 5cm] {$3w$};
\node (time4) [below of = time3, node distance = 7.3cm] {$4w$};
\node (time) [above of = time0, node distance = 0.5cm] {\textbf{Timeline}};

\draw [arrow] (time0) -- (time1);
\draw [arrow] (time1) -- (time2);
\draw [arrow] (time2) -- (time3);
\draw [arrow] (time3) -- (time4);

\end{tikzpicture}
\caption{Flowchart for Alice and Bob during each \round. Here $w$ denotes the word length in the current round.}
\label{fig:statediagram}
\end{center}
\end{figure}

\begin{algorithm}[H]
	\caption{Message exchange algorithm for the sender. {\tt Send-Message} is only called at the beginning of a round.}
	\label{aliceAlgo}
	\begin{algorithmic}[1]
		\Procedure{\sendMessage}{$m$} 
		\LineComment{$b$ is a persistent variable for the parity bit. On the first call to \sendMessage, $b$ is set}
		\LineComment{to $0$. On subsequent calls it is whatever it was set to on the previous call.}
		\While{\texttt{true}}
			\State Generate random key $k_A$ of length $\kappa_r$ \Comment{$r$ is the current round number.}.
			\State \label{step} Send $\encodedMsg{\texttt{KEY?}}{k_A}{r}$.
			\State \label{alg:AliceReadsBobsKey} $M_1 \gets w_r$ bits on the channel from the receiver.
			\If{$\silence{M_1}$}		\Comment{assume the receiver has already terminated.}
				\State \label{alg:AliceReturnsBecauseSilence}Stay silent for $2w_r$ time steps and \Return
			\Else 
				\State $(x,k) \gets \decodeMsg{M_1}{t}$
				\If {$k \neq k_A$} 
				\State Stay silent for $2w_r$ time steps.
			\Else
				\State $k_B \gets x$
				\State \label{alg:AliceKey} Send $\encodedMsg{(m,b)}{k_B}{t}$. 
				\State $M_2 \gets w_r$ bits on the channel from the receiver.
				\If{$\silence{M_2}$}
					\State \label{alg:AliceReturnsFinally} $b \gets \lnot b$ and \Return
				\EndIf
			\EndIf
			\EndIf
		\EndWhile
		\EndProcedure
	\end{algorithmic}
\end{algorithm}

\begin{algorithm}[H]
	\caption{Message exchange algorithm for the receiver.  Receive-Message is only called at the beginning of a round.}
	\label{bobAlgo}
	\begin{algorithmic}[1]
		\Procedure{\receiveMessage}{\ }
		\LineComment{$\hat{b}$ is a persistent variable for the parity bit. On the first call to \receiveMessage, $\hat{b}$ is }
		\LineComment{set to $0$. On subsequent calls it is whatever it was set to on the previous call.}
		\State $M_1' \gets w_r\ $bits on the channel from the sender. 
		\If {$\silence{M_1'}$}
			\State Stay silent for $3w_r$ time steps.  \Comment{$r$ is the current round number.}
		\Else
			\State $(x',k') \gets \decodeMsg{M_1'}{t}$ 
			\If {$x' \neq \texttt{KEY?}$}
			\State Send noise for $w_r$ time steps.
			\State Stay silent for $2w_r$ time steps.
		\Else
			\State $k_A \gets k'$. 
			\State Generate random key $k_B$ of length $\kappa_t$.
			\State \label{alg:BobKey} Send $\encodedMsg{k_B}{k_A}{t}$.
			\State $M_2' \gets w_r\ $bits on the channel from the sender.
			\State $(x'',k'') \gets \decodeMsg{M_2'}{t}$
			\If{$k'' \neq k_B$}
				\State \label{alg:bobNoise} Send noise for $w_r$ time steps.
			\Else
				\State \label{alg:bobDecode} $(m',b') \gets x''$
				\If {$b' \neq \hat{b}$ different from last message}
					\State \label{bob:record} Set $\hat{b} \gets b'$ and record the message $m'$ from the sender.
				\EndIf
				\State Stay silent for $w_r$ time steps.
			\EndIf
			\EndIf
		\EndIf
		\EndProcedure
	\end{algorithmic}
\end{algorithm}

\subsection{The protocol $\pi'$}

Our main protocol $\pi'$ is presented as Algorithm~\ref{multipartyAlgo}.  We make use of the definition of an I/O automaton (see~\cite{lynch1996distributed} Chapter 14.1.1) to represent $\pi$.  We assume that for each user $u$ in the network, $\pi$ provides an I/O automaton $\pi_u$ with the following properties.

\begin{itemize}
	\item $\pi_u$ has a single \emph{initial} state
	\item $\pi_u$ has some subset of states that are \emph{termination} states.  Each termination state may have a value for $u$ to \emph{output}.
	\item There is a set of \emph{transition relations}, each from one state to another state, each labeled with an \emph{action}, where this action may be either an \emph{input} action (e.g. receiving a message) or an \emph{output} action (e.g. sending a message).  These transitions satisfy the property for every state $s$, for every possible input action, $a$, there is a transition from $s$ to some other state that is labelled with $a$.
\end{itemize}

\begin{algorithm}[H]
	\caption{Protocol $\pi'$, run at each node $u$}
	\label{multipartyAlgo}
	\begin{algorithmic}[1]
		\State $s \leftarrow$ initial state of $\pi_u$ 
		\While{$s$ is not a termination state}
		\If{there is a transition relation from state $s$ to state $s'$ in $\pi_u$ labeled with an output action to send message $m$ to neighbor $v$}
				\State Schedule $\Call{\sendMessage}{m}$ to run on the channel to $v$. 
\par\quad If no call to $\sendMessage$ is currently running on the channel, begin immediately. \par \quad Otherwise, set up the call to begin as soon as all currently scheduled $\sendMessage$ 
\par\quad calls on this channel have finished running.
				\State Transition to state $s'$ in $\pi_u$; $s \leftarrow s'$.
			\Else
			\Repeat
			\State for each neighbor $v$, in parallel run $\Call{\receiveMessage}$ on the channel from $v$.
			\Until {the $\receiveMessage$ calls have recorded some non-empty set $S$ of messages}
			\For{each message $m$ in $S$}
				\State Let $s'$ be the target state on the transition relation from $s$ with input action $m$.
				\State Transition to state $s'$ in $\pi_u$; $s \leftarrow s'$. 
			\EndFor
			\EndIf
		\EndWhile
		\State \label{alg:BobTerminationState}Set processor $u$'s output based on the termination state.
		\State Continue executing any remaining scheduled $\sendMessage$ calls until all have returned.
	\end{algorithmic}
\end{algorithm}  

\section{Analysis}
\label{sec:unboundedTanal}

We now analyze the main algorithm as presented in Section~\ref{sec:unboundedTalgo}. We begin by computing the failure probability for the algorithm, by considering the three bad events as before, and then take a union bound over all the \rounds. We will then prove that the algorithm is correct and terminates in finitely many time steps. Finally, we compute an upper bound on the expected number of bits the algorithm sends. 

\subsection{Probability of Failure}
\label{sec:probFailure}
We define three bad events for a round.
\begin{enumerate}
	\item \emph{\failAMD} : The adversary is able to flip the bits to convert the message into another valid word. In this case, the affected users end their \rounds either with authentication failures or with knowledge of bits which are not in $\pi$.
	\item \emph{\failConversionSilence} : The adversary is able to flip bits in such a way that some user's random bits look like silence to his neighbor, resulting in the latter ending his \round without the knowledge of this failure.
	\item \emph{\failKeyInstallation} : The adversary installs a correct key from Bob when Alice is silent (to successfully simulate line~\ref{alg:AliceKey} of Algorithm~\ref{aliceAlgo}) or installs a correct key from Alice when Bob is silent (to successfully simulate line~\ref{alg:BobKey} of Algorithm~\ref{bobAlgo}).
\end{enumerate}
We now bound the probabilities of each of these failure events to eventually prove that our algorithm succeeds with probability $1-\errorprob$.

\begin{lemma}
	\label{lem:AMDfailureunbounded}
	In $\pi'$, \emph{\failAMD} occurs with probability at most $\errorprob/3$.
\end{lemma}
\begin{proof}
	Recall that in any round, four words are exchanged between any pair of users, and that at most ${n \choose 2} < n^2$ pairs of users may be exchanging words with each other. Furthermore, in a given round $r$, the AMD failure probability for a single word is set to be at most $\amdstr_r = \paran{n \pi r}^{-2} \errorprob/2$, as discussed in Section~\ref{sec:overview}. \tom{Can we make that a more precise lemma to cite?  Section 2.2 is pretty big.} Now, for $1 \leq k \leq 4, 1 \leq i\neq j \leq n$ and $r \geq 1$, define $\xi_{i,j,k,r}$ to be the event that in round $r$, AMD failure occurs in the $k^{th}$ word exchanged between users $i$ and $j$. Then, $\Pr( \xi_{i,j,k,r} ) \leq \paran{n \pi r}^{-2} \errorprob/2$. Hence, by a union bound, we get 
\[ 
\Pr \left( \bigcup_
{ i,j,k,r}
\xi_{i,j,k,r} \right) \leq \sum_{k=1}^4 \sum_{\substack{i,j = 1 \\ i \neq j}}^n \sum_{r \geq 1} \frac{\errorprob}{2(n \pi r)^2} 
\le \frac{2 \errorprob}{\pi^2} \sum_{r \ge 1} \frac{1}{r^2}
= \frac{\errorprob}{3}. 
\]
\end{proof}

\begin{lemma}
		\label{lem:chernoff}
		For $b\geq 95$, the probability that a $b$-bit string sampled uniformly at random from $\{0,1 \}^b$ has fewer than $b/3$ bit alternations is at most $e^{-b/19}$.
	\end{lemma}
	\begin{proof}
		Let $s$ be a string sampled uniformly at random from $\{ 0,1\}^b$, where $b \geq 95$. Denote by $s[i]$ the $i^{th}$ bit of $s$. Let $X_i$ be the indicator random variable for the event that $s[i]\neq s[i+1]$, for $1 \leq i < b$. Note that all $X_i$'s are mutually independent. Let $X$ be the number of bit alternations in $s$. Clearly, $X = \sum_{i=1}^{b-1}X_i$, which gives $\mathbb{E}(X) = \sum_{i=1}^{b-1}\mathbb{E}(X_i)$, using the linearity of expectation. Since $\mathbb{E}(X_i)=1/2$ for all $1 \leq i < b$, we get $\mathbb{E}(X) = (b-1)/2$. Using a multiplicative version of Chernoff's bound (see ~\cite{dubhashi2009concentration}), we have that for $0 \leq t \leq \sqrt{b-1}$,
		\begin{equation*}
			\Pr \left( X < \frac{b-1}{2} - \frac{t\sqrt{b-1}}{2} \right) \leq e^{-t^2 /2}.
		\end{equation*}
		To obtain $\Pr (X < b/3 )$, we set $t = \frac{b-3}{3\sqrt{b-1}}$ to get
		\begin{equation*}
		\Pr( X < b/3) \leq e^{-\frac{(b-3)^2}{18(b-1)}} \leq e^{- b/19} \quad \text{for $b \geq 95$},
		\end{equation*}
where the condition $b \ge 95$ comes from rounding the solution to a quadratic equation.
\end{proof} 

\begin{lemma}
	\label{lem:Silenceunbounded}
	In $\pi'$, \emph{\failConversionSilence} occurs with probability at most $\errorprob/3$.
\end{lemma}
\begin{proof}
	Recall that in round $r$, each encoded message includes $38 \left \lceil \log \paran{\frac{2n \pi r}{\sqrt{\errorprob}}} \right \rceil$ random bits at the end. Then, Lemma~\ref{lem:chernoff} tells us that the probability that adversary is able flip these random bits of a word to forge silence is at most $2^{-\frac{38 \left \lceil \log \paran{\frac{2n \pi r}{\sqrt{\errorprob}}} \right \rceil}{19}} \leq \paran{n \pi r}^{-2} \errorprob/4$. Thus, similar to Lemma~\ref{lem:AMDfailureunbounded}, for $1 \leq k \leq 4, 1 \leq i\neq j \leq n$ and $r \geq 1$, define $\xi_{i,j,k,r}$ to be the event that in round $r$, the $k^{th}$ word exchanged between users $i$ and $j$ is converted to silence, so that $\Pr (\xi_{i,j,k,r}) \leq \paran{n \pi r}^{-2} \errorprob/4$. Hence, by using a similar union bound as Lemma~\ref{lem:AMDfailureunbounded}, we get the desired bound.
\end{proof}

\begin{lemma}
	\label{lem:guessUnbounded}
	In $\pi'$, the adversary is able to guess the key of some user with probability at most $\errorprob/3$.
	\end{lemma}
\begin{proof}
	Recall that in round $r$, keys are of length $2 \left \lceil \log \paran{\frac{4n \pi r}{\sqrt{\errorprob}}} \right \rceil$. Each such key is generated uniformly at random from the set of all binary strings of this length. Thus, the probability of guessing this key is at most $2^{-2 \left \lceil \log \paran{\frac{4n \pi r}{\sqrt{\errorprob}}} \right \rceil} \leq \paran{n \pi r}^{-2} \errorprob/16$. Again, similar to Lemma~\ref{lem:AMDfailureunbounded}, for $1 \leq k \leq 4, 1 \leq i\neq j \leq n$ and $r \geq 1$, define $\xi_{i,j,k,r}$ to be the event that in round $r$, the adversary is able to guess the key in the $k^{th}$ word exchanged between users $i$ and $j$, so that $\Pr (\xi_{i,j,k,r}) \leq \paran{n \pi r}^{-2} \errorprob/16$. Hence, by using a similar union bound as Lemma~\ref{lem:AMDfailureunbounded}, we get the desired bound.
	\end{proof}

\begin{lemma}
\label{lem:failureTotal}
	With probability at least $1-\errorprob$, none of the failure events happen during a run of $\pi'$.
\end{lemma}
\begin{proof}
	A run of protocol $\pi'$ fails if any of the three failure events described above happen. From Lemmas~\ref{lem:AMDfailureunbounded},~\ref{lem:Silenceunbounded} and~\ref{lem:guessUnbounded}, the total probability of failure is computed by using union bound over the three failure events, which gives $\errorprob/3 + \errorprob/3 + \errorprob/3 = \errorprob$. Hence, the run succeeds with probability at least $1 - \errorprob$.
\end{proof}

\subsection{Correctness}
All lemmas in this section assume that none of the failure events occur. Without loss of generality, we also assume that Alice and Bob generate their keys in every round. We use the phrase \terminatedPi to mean finished executing line~\ref{alg:BobTerminationState} of Algorithm~\ref{multipartyAlgo}.

\begin{lemma}
	\label{lem:goodEvents}
	Fix a round $r$ of $\pi'$. Let $W$ be any of the four sequences of $w_r$ bits onthe channel from Alice to Bob or the channel from Bob to Alice in this round and $k_A^{(r)}$, $k_B^{(r)}$ be the keys that Alice and Bob generate in this round. Then exactly one of the following hold.
	\begin{enumerate}
		\item $\silence{W}$ is true, in which case the sender on the channel was silent as well.
		\item $\decodeMsg{W}{t}=\mathcal{C} \neq \perp$, where $t$ is the time step at which the round $r$ began. Then,
			\begin{enumerate}
				\item If $\mathcal{C} = (\texttt{KEY?},k)$, then either Alice sent $W$ with key $k = k_A^{(r)}$, or Alice is silent and the adversary sent $W$.
				\item If $\mathcal{C} = (k',k_A^{(r)})$, then Bob has not \terminatedPi yet, and he sent $W$ with key $k' = k_B^{(r)}$.
				\item If $\mathcal{C} = (x,k_B^{(r)})$, then Alice sent $W$ with content $x = m^{(r)}$, where $m^{(r)}$ is Alice's message for Bob in round $r$.
				\item Otherwise, the adversary sent $W$ on a silent channel. 
			\end{enumerate} 
			\item Bob executed line~\ref{alg:bobNoise} of Algorithm~\ref{bobAlgo}.
			\item $W$ is the outcome of adversarial tampering on the channel.
	\end{enumerate}
\end{lemma}
\begin{proof}
	In any given round $r$, whenever Alice or Bob send a word to the other, this word does not convert to silence because we assume that the failure event \failConversionSilence does not occur. Hence, if silence is received on the channel, then the sender on the channel must have been silent at that time. This proves part (1) of our lemma. \\
	
	If, however, a valid AMD codeword is received, which decodes into $\mathcal{C}$, the following cases are possible. Case 2(a) : If $\mathcal{C} = (\texttt{KEY?},k)$, then either Alice issued this key request to Bob with her key $k_A^{(r)}$ (line~\ref{step} of Algorithm~\ref{aliceAlgo}), or she was silent. In the former case, the adversary is unable to put a key $k \neq k_A^{(r)}$ in the codeword since the failure event \failAMD does not occur, and hence, $k = k_A^{(r)}$. In the latter case, since the channel is silent, the adversary must have issued the key request with some key $k$. Case 2(b) : If $\mathcal{C} = (k',k_A^{(r)})$, then since we assume that the failure event \failKeyInstallation does not occur, it must be the case that Bob sent $W$ and not the adversary. Thus, Bob must have included his own key $k_B^{(r)}$ to the codeword along with a copy of Alice's key (line~\ref{alg:BobKey} of Algoithm~\ref{bobAlgo}). This is only possible when Bob has not \terminatedPi yet. Case 2(c) : If $\mathcal{C} = (x,k_B^{(r)})$, then Alice must have sent $W$ since the the failure event \failKeyInstallation does not occur. Also, since we assume that the failure event \failAMD does not occur, the adversary would have not been able to convert $m^{(r)}$ into another message successfully, and hence, $W$ must contain the message that Alice has for Bob in this round. Case 2(d) : If the AMD codeword is neither of the above three cases, then it must be the case that the sender was on the channel at the time $W$ was on the channel, since \failAMD does not occur. Hence, the adversary must have sent $W$ on a silent channel.\\
	
	If neither silence nor an AMD codeword is received, then $W$ is noise. The only step in $\pi'$ where noise is intentionally put on the channel is when Bob has to inform Alice that he did not receive her message correctly (line~\ref{alg:bobNoise} of Algorithm~\ref{bobAlgo}). Thus, if Bob sent this noise, Case 3 of our lemma holds, else Case 4 must hold where the adversary has tampered with the bits on the channel so that $W$ becomes noise.
\end{proof}

\begin{lemma}
\label{lem:msgDel}
Assume Alice calls \sendMessage(m) in some round $r_1$ for some message $m$ and bit $b$.  Then the following hold: (1) Alice returns from \sendMessage(m) in some round $r_3 \geq r_1$; and (2) either Bob records the message $m$ (line~\ref{bob:record} of Algorithm 2) in exactly one round $r_2$ where $r_1 \leq r_2 \leq r_3$, or Bob does not record the message $m$ between rounds $r_1$ and $r_3$ because Bob \terminatedPi.
\end{lemma}

\begin{proof}
We first show that if Alice calls \sendMessage(m) in round $r_1$, then there must exist some round $r_3 \geq r_1$ in which this call returns.  
Since the adversary's budget is finite, there must be some round, $r_4$, after which no bits are ever flipped.  
If the call returns before round $r_4$, then part (1) is proven, so we only consider the cases where the call extends past round $r_4$. In round $r_4+1$, after Alice has sent a key request she either hears silence, indicating that Bob has \terminatedPi, and the call returns, or she correctly receives Bob's key, and uses it to send $m$ which Bob correctly receives, and records (since the bit $b$ is different than the bit in the last message recorded).  Next, Alice hears silence from Bob.  This ends the call. Either way the call returns at the end of round $r_4 +1$.  Thus, in every case there is some round $r_3 \geq r_1$ in which the call to \sendMessage(m) returns. \\

We now prove part (2) of our lemma.  We first show that the message $m$ is recorded at most once by Bob in the rounds $r_1$ to $r_3$.  By Lemma~\ref{lem:goodEvents} (2(c)), the bit $b'$ received by Bob in line~\ref{alg:bobDecode} of Algorithm 2 must be the same as the bit $b$ sent by Alice from rounds $r_1$ to $r_3$.  Since this bit never changes, Bob will record a message at most once in rounds $r_1$ to $r_3$. \\

If Bob \terminatedPi\ before round $r_3$, then part (2) of our lemma statement does not require Bob to record the message $m$, and so that part or our lemma is trivially true.  Thus, for the remainder of the proof, we assume that Bob has not \terminatedPi\ before round $r_3$. \\ 

Consider round $r_3$ in which the call to \sendMessage(m) returns. Let $M_1$ be the string read by Alice on line~\ref{alg:AliceReadsBobsKey} of Algorithm~\ref{aliceAlgo}. Since Bob has not terminated, Lemma~\ref{lem:goodEvents} (2(b)) guarantees that Bob sent $M_1$, and therefore the call to Algorithm~\ref{bobAlgo} does not return on line~\ref{alg:AliceReturnsBecauseSilence}. Hence, it returns on line~\ref{alg:AliceReturnsFinally}.  Since this is the last round, $M_1$ must correctly decode to $(k_B, k_A)$ in Algorithm~\ref{aliceAlgo}. Thus, Alice sends Bob $m$ using $k_B$ and hears silence subsequently. Thus, Bob must have actually been silent at this time (Lemma~\ref{lem:goodEvents} (1)), which only happens if he has either now or previously recorded $m$. Hence, Bob must have recorded the message $m$ in some round $r_2 \leq r_3$.
\end{proof}

\begin{lemma}
\label{lem:correctnessRounds}
The following holds for $\pi'$. For any message $m$ that Bob records (line~\ref{bob:record} of Algorithm~\ref{bobAlgo}) in some round $r_2$, Alice started a call to \sendMessage \ with the message $m$, in a round $r_1 \leq r_2$ and returned from that call in some round $r_3 \geq r_2$. 	
\end{lemma}
\begin{proof}
Consider the round $r_2$ in which the Bob records a message $m$ (line~\ref{bob:record} of Algorithm 2). In round $r_2$, in line~\ref{alg:bobDecode} of Algorithm~\ref{bobAlgo}, let $\decodeMsg{M_1'}{t} = ((m,b),k_B)$. Thus, by Lemma~\ref{lem:goodEvents} (2(c)), Alice must have sent $\encodedMsg{(m,b)}{k_B}{t}$, during the transmission of $M_1'$.  Hence, Alice must be in a call to \sendMessage \ with the message $m$, and this call must have begun in some round $r_1 \leq r_2$. \\

Finally, we know by Lemma~\ref{lem:msgDel}(1) that every call to \sendMessage \ ends in some round $r_3$.  Since, by the above, Alice is in the call during round $r_2$, it must be the case that Alice returns from the call in some round $r_3 \geq r_2$.
\end{proof}

\begin{lemma}
	\label{lem:correctness}
	Algorithm $\pi'$ terminates with a correct simulation of an asynchronous run of $\pi$. 
\end{lemma}
\begin{proof}

We show that for every pair of users $u$ and $v$, protocol $\pi'$ correctly simulates a FIFO message channel(see~\cite{lynch1996distributed} Chapter 14.1.2) from $u$ to $v$ during the simulation of $\pi$.  Then a
direct induction shows that for each processor $u$, $\pi'_u$ simulates $\pi_u$ correctly.  \\

Fix an arbitrary channel from $u$ to $v$ in the network.   Let $Q_{u,v}$ be the queue of $\sendMessage$ procedures that are maintained in $\pi'$ by  $u$ of messages to send to  $v$.  

We require the following facts:
\begin{enumerate}
	\item In any time step, there is a transition in $\pi_u$ across a transition relation with an output action to send a message $m$ to user $v$, if and only if the procedure $\sendMessage$ for message $m$ to user $v$ is pushed on the queue $Q_{u,v}$ in that time step.
	\item Every $\sendMessage$ procedure for message $m$ on $Q_{u,v}$ will eventually start at some round $r_1$ and end in some round $r_3$.  Moreover, user $v$ transitions across an input transition relation for message $m$ from user $u$ at most once in some round $r_2$, $r_1 \leq r_2 \leq r_3$.  Moreover, if there is no such input transition relation, than $\pi_u$ has entered a termination state before round $r_3$.
\tom{The notion of ``transitioning across an input transition relation,'' hasn't been discussed before or properly defined.  Transition relations are mentioned on page 13 as coming from the definition of I/O automata, but not explained.}
	\item For any transition along a transition relation in $\pi_v$ with an input action, in some round $r_2$, user $u$ started a call to $\sendMessage$ with the message $m$, in a round $r_1 \leq r_2$ and returned from that call in some round $r_3 \geq r_2$.
\end{enumerate}

Fact (1) follows directly from Algorithm~\ref{multipartyAlgo} steps 3-6.  The first sentence of Fact (2) follows by induction and Lemma~\ref{lem:msgDel}, and the remainder of the fact follows directly from Lemma~\ref{lem:msgDel}.  Fact (3) follows from Lemma~\ref{lem:correctnessRounds}.\\

Together, the facts show that no matter what the actions of the adversary, the protocol $\pi'$ correctly simulates a FIFO message channel from user u to user v.  In particular, we have: 1) when a transition is taken in $\pi_u$ with output action to send message $m$ to user $v$, this message is put on a queue; 2) all transitions in $\pi_v$ with an input action to receive a message $m$ from user $u$ are associated with the removal of message $m$ from the queue; and 3) all messages are eventually removed from the queue, triggering transitions across transition relations with input actions in $\pi_v$, unless $\pi_v$ is already in a termination state.
\end{proof}

\subsection{Resource Costs}
We now compute the expected number of bits sent and the latency of $\pi'$. Let $\tau(r)$ be the time step at which the $r^{th}$-round of our algorithm begins. Clearly, $\tau(1) = 1$.   

\begin{lemma}
	\label{lem:taubound}
	In $\pi'$, round $r \geq 1$ begins at time step $\tau(r) = \Theta\paran{r \log \paran{\frac{nr}{\errorprob}}}$.
	\end{lemma}
\begin{proof}
For all $r \geq 1$, note that round $r+1$ begins as soon as the number of time steps corresponding to four words of the round $r$ have passed. Hence, we can compute $\tau(r)$ using the recurrence $\tau(r) = \tau(r-1) + 4w_{r-1}$, where $\tau(1) = 1$. This gives $\tau(r) = \tau(1) + 4\sum_{i=1}^{r-1} w_i$. Now, from equation~(\ref{eq:wordLength}), using $w_i \leq C_1 \log \paran{\frac{ni}{\errorprob}} + C_2$, we get $\tau(r) \leq 1 + 4\sum_{i=1}^{r-1} \paran{C_1 \log \paran{\frac{ni}{\errorprob}} + C_2}$, which gives $\tau(r) = \Theta \paran{r \log \paran{\frac{nr}{\errorprob}}}$.
\end{proof}

\begin{lemma}
	\label{lem:numRoundsTil}
	If $\tau(x) \leq z$ for some $x \geq 1$ and $z \geq 1$, then $x= \mathcal{O}(z/ \log \paran{\frac{nz}{\errorprob}})$.
\end{lemma}

\begin{proof}
	We prove the bound on $x$ for the case that $\tau(x) = z$.  By Lemma~\ref{lem:taubound}, we know that $\tau(x) = C \paran{x \log \paran{\frac{nx}{\errorprob}}}$, for some constant $C$.  Thus, $z = C x \log \paran{\frac{nx}{\errorprob}}$, and we get (*) $x \leq z / (C \log \paran{\frac{nx}{\errorprob}})$.
	
Note that $z = C x \log \paran{\frac{nx}{\errorprob}} \leq C x \log x$.  Taking logs of both sides, we get that $\log z \leq C' \log x$ for some constant $C'$, which implies that $\log \paran{\frac{nz}{\errorprob}} \leq C' \log \paran{\frac{nx}{\errorprob}}$.  Now plugging this back into (*), we get that $x \leq C'' z/ \log \paran{\frac{nz}{\errorprob}}$ for some constant $C''$.
\end{proof}

\begin{lemma}
	\label{lem:resourceCosts1}
	If $\pi'$ succeeds, then it has the following resource costs.
	\begin{itemize}
	\item The number of bits sent is $\mathcal{O} \paran{L\log \paran{\frac{n L}{\errorprob}} + T}$.
	\item The latency is $\mathcal{O} \paran{\underset{p}{\max} \left \{ \Lambda_p\log \paran{\frac{n(L+T)}{\errorprob}} + T_p \right \}}$, where 
          the maximum is taken over all communication paths, $p$, in the asynchronous simulation of $\pi$, $\Lambda_p$ is the latency of $p$, and $T_p$ is the     total number of bits flipped by the adversary on edges in $p$.
	\end{itemize}
\end{lemma}
\begin{proof}

To bound the number of bits sent, we assume pessimistically that in every round of $\pi'$, there is an attempt to send exactly one message.  This maximizes the number of bits sent since word sizes increase with time.  

Since each word in $\pi'$ is ECC encoded, the adversary must flip a constant fraction of the bits to successfully corrupt the word, and thereby compromise the round.  Let $x$ be the number of rounds in which some word was successfully corrupted by the adversary. Since the length of the words increases with successive rounds, $T$ must at least be a constant, $C$, times the number of bits in the first $x$ rounds of $\pi'$, and hence, we must have $\tau(x) \leq (T+1)/C$. 

By Lemma~\ref{lem:numRoundsTil}, we know that $x = \mathcal{O}\paran{\frac{T}{\log \paran{n(T+1)/\errorprob}}}$.  Thus, since $\pi'$ requires $L$ progressive rounds, the total number of rounds is $r = L + \mathcal{O}\paran{\frac{T}{\log \paran{n(T+1)/\errorprob}}}$.

Hence the total number of bits sent is at most $\tau(r)$.  By Lemma~\ref{lem:taubound},

\begin{eqnarray*}
	\tau(r) & = & \mathcal{O}\paran{\paran{L + \frac{T}{\log \paran{n(T+1)/\errorprob}}} \log \paran{\frac{n}{\errorprob}\paran{L + \frac{T}{\log \paran{n(T+1)/\errorprob}}}}} \\
	& \leq & \mathcal{O}\paran{\paran{L + \frac{T}{\log \paran{n(T+1)/\errorprob}}} \log \frac{n(L+T)}{\errorprob}} \\
        & = & \mathcal{O} \paran{L\log \paran{\frac{n (L+T)}{\errorprob}} + T} \\
	& = & \mathcal{O} \paran{L\log \paran{\frac{n L}{\errorprob}} + T}.
\end{eqnarray*}

The second line above follows from the fact that $\frac{T}{\log \paran{n(T+1)/\errorprob}} \leq T$, and the third line from the fact that if $L = O(T)$, $\log (L+T) = O(\log T)$.  The final line above follows fro the fact that $\log(L + T) = \log(L) + \log(1 + T/L) \le \log(L) + T/L$, and hence
$L \log(L+T) \le L \log(L) + T$.
This bounds the number of bits sent.

To bound the latency, we note that the argument above holds for any 
communication path $p$ in the asynchronous simulation of $\pi$.
For the $p$ which achieves the maximum, it follows by induction that all other required messages will have already
been received by the time they are needed, and so $p$ determines the overall latency.
\tom{Added a sentence to try to flesh this out a bit better.}
\end{proof}

Now we are ready to complete the proof of Theorem~\ref{thm:mainUnbounded}, in the case when the messages of $\pi$ are all single bits.

\begin{proof}[Proof of Theorem~\ref{thm:mainUnbounded}, special case]
By Lemma~\ref{lem:failureTotal}, no failure event happens with probability at least $1-\errorprob$, and by Lemma~\ref{lem:correctness}, in such a situation, $\pi'$ terminates with an asynchronous simulation of $\pi$.  Upon correct termination, the resource cost bounds hold by Lemma~\ref{lem:resourceCosts1}. 
\end{proof}



\section{Simulating Protocols with Messages of Arbitrary Length}
\label{sec:alpha}

\jared{These next two paragraphs should be revised to match more closely with the model from H.'s paper giving an algorithm that works for variable message lengths.}

So far, we have assumed that all messages sent in the protocol $\pi$ are single bit messages. In fact $\pi$ could send messages of a fixed constant size, or even of size $\log n$ and our results would still hold, provided the initial word length was chosen long enough to be able to encode the messages in $\pi$, the limiting factor here being that the encoding function $\encodedMsg{}{}{t}$ can only encode strings up to some length. In this section, we show how to go beyond this restriction and modify our algorithm to be able to simulate protocols $\pi$ of arbitrary and variable message length, and we show that the cost of the simulation scales well relative to the average message length in $\pi$. \\

Before we proceed, we need to discuss what it means for $\pi$ to have variable message length. After all, even in the noise-free setting, if Alice is to send Bob a previously unspecified number of bits, then how is Bob to know when to stop reading bits? \varsha{More philosophical discussion needed here: Why not End-of-Message symbol which is implicitly standard in the study of asynch algs. Need to know how such a symbol would be corrupted.  If uncorruptible then makes everything too easy, uninteresting. Can hide info in that etc. }

Here we assume that the messages from Alice to Bob in $\pi$ come from a fixed prefix-free language $\pfl$ = $\pfl_{A,B}$. Recall that a language $\pfl$ is prefix-free if for any pair of strings $s$ and $s'$ in $\pfl$ neither is a prefix of the other.  Note that different prefix-free languages may be used between each pair of players.

\varsha{More discussion? Examples of PFLs?}

\subsection{Algorithm}

The main simulating algorithm $\pi'$ (Algorithm~\ref{multipartyAlgo}) remains unchanged in this setting. Only the sending and receiving algorithms need to change to reflect the fact that the message to be sent may be longer than a single call to $\encodedMsg{}{}{t}$ can support. In Algorithms~\ref{aliceAlgo-vl} and~\ref{bobAlgo-vl} below, we highlight the necessary changes in red.

\begin{algorithm}[H]
	\caption{Message exchange algorithm for the sender.}
	\label{aliceAlgo-vl}
	\begin{algorithmic}[1]
		\Procedure{\sendMessage}{$m$}
		\LineComment{\hlt{$b$ is a persistent variable for the parity bit. On the first call to \sendMessage, $b$ is set}}
		\LineComment{\hlt{to $0$.  On subsequent calls it is whatever it was set to on the previous call.}}
	\State \hlt{$j \gets 0$}
		\While{\hlt{$j < |m|$}
			\State $(w_r,\kappa_t) \gets$ \Call{Word-Params}{$t$} \Comment{$t$ is the current time step.} 
			\State Generate random key $k_A$ of length $\kappa_t$.
			\State \label{step-vl} Send $\encodedMsg{\texttt{KEY?}}{k_A}{t}$.
			\State \label{alg:AliceReadsBobsKey-vl} $M_1 \gets w_r$ bits on the channel from the receiver.
			\If{$\silence{M_1}$}		\Comment{Assume the receiver has already terminated.}
				\State $b \gets \lnot b$
				\State \label{alg:AliceReturnsBecauseSilence-vl} Stay silent for $2w_r$ time steps and \Return
			\Else 
				\State $(x,k) \gets \decodeMsg{M_1}{t}$
				\If {$k \neq k_A$} 
				\State Stay silent for $2w_r$ time steps.
			\Else
				\State $k_B \gets x$
\State \hlt{$M \gets m[j,j+\kappa_t]$ \Comment{Next $\kappa_t$ bits of $m$} }
				\State \label{alg:AliceKey-vl} \hlt{Send $\encodedMsg{(M, b)} {k_B}{t}$.} 
				\State $M_2 \gets w_r$ bits on the channel from the receiver.
				\If{$\silence{M_2}$}
					\State \hlt{$j \gets j+\kappa_t$}
					\State \hlt{$b \gets \lnot b$}	
				\EndIf
			\EndIf
			\EndIf
			\State \label{alg:AliceReturnsFinally-vl} \hlt{\Return}
		\EndWhile
		\EndProcedure
	\end{algorithmic}
\end{algorithm} 

\begin{algorithm}[H]
	\caption{Message exchange algorithm for the receiver.}
	\label{bobAlgo-vl}
	\begin{algorithmic}[1]
		\Procedure{\receiveMessage}{\ } 
		\LineComment{\hlt{$\hat{b}, \mu, \lambda$ are persistent variables for the parity bit, the partially received message, and length}}
		\LineComment{\hlt{of the recorded message, respectively. On the first call to \receiveMessage, $\hat{b} \gets 0$,}}
		\LineComment{\hlt{$ \mu \gets \varnothing, \lambda \gets 0$. On subsequent calls these variables are whatever they were set to on the}}
		\LineComment{\hlt{previous call.}}
		\State \label{Bobstep-vl} $(w_r,\kappa_t) \gets$ \Call{Word-Params}{$t$}	\Comment{$t$ is the current time step.} 
		\State $M_1' \gets w_r\ $bits on the channel from the sender. 
		\If {$\silence{M_1'}$}
			\State Stay silent for $3w_r$ time steps.
		\Else
			\State $(x',k') \gets \decodeMsg{M_1'}{t}$ 
			\If {$x' \neq \texttt{KEY?}$}
				\State Send noise for $w_r$ time steps.
				\State Stay silent for $2w_r$ time steps.
			\Else
				\State $k_A \gets k'$. 
				\State Generate random key $k_B$ of length $\kappa_t$.
				\State \label{alg:BobKey-vl} Send $\encodedMsg{k_B}{k_A}{t}$.
				\State $M_2' \gets w_r$ bits on the channel from the sender.
				\State $(x'',k'') \gets \decodeMsg{M_2'}{t}$
				\If{$k'' \neq k_B$}
					\State \label{alg:bobNoise-vl} Send noise for $w_r$ time steps.
				\Else
					\State \label{alg:bobDecode-vl} $(M', b') \gets x''$
					\If{\hlt{$\mu = \varnothing$}}
						\If{\hlt{$b' \neq \hat{b}$}}
							\State \hlt{$\lambda \gets |M'|$}
							\State \hlt{$\mu \gets M'$}
						\EndIf
					\Else
						\If{\hlt{$b' \neq \hat{b}$}}
							\State \hlt{$\lambda \gets |M'|$}
							\State \hlt{Append $M'$ to $\mu$.}
						\Else 
							\State \hlt{Replace last $\lambda$ bits of $\mu$ with $M'$.}
							\State \hlt{$\lambda \gets |M'|$}}
						\EndIf
					\EndIf
					\If{\hlt{$\mu$ is a completed message of $\pfl$}}
						\State \label{bob:record-vl}\hlt{ Record the message $\mu$ from the sender.}
						\State \hlt{$\mu \gets \varnothing$}
						\State \hlt{$\hat{b} \gets b'$}
					\EndIf
					\State Stay silent for $w_r$ time steps.
				\EndIf
			\EndIf
		\EndIf
		\EndProcedure
	\end{algorithmic}
\end{algorithm}

\section{Analysis}
\label{sec:unboundedTanalArbLength}

\subsection{Correctness}

\begin{lemma}
\label{lem:msgDel-vl}
Assume Alice calls \sendMessage(m) in some round $r_1$ for some message $m$ and bit $b$.  Then the following hold: (1) Alice returns from \sendMessage(m) in some round $r_3 \geq r_1$; and (2) either Bob records the message $m$ (line~\ref{bob:record-vl} of Algorithm~\ref{bobAlgo-vl}) in exactly one round $r_2$ where $r_1 \leq r_2 \leq r_3$, or Bob does not record the message $m$ between rounds $r_1$ and $r_3$ because he \terminatedPi.
\end{lemma}

\begin{proof}
We first show that if Alice calls \sendMessage(m) in round $r_1$, then there must exist some round $r_3 \geq r_1$ in which this call returns.  
Since the adversary's budget is finite, there must be some round, $r_4$, after which no bits are ever flipped.  
If the call returns before or during round $r_4$, then part (1) is proven, so we only consider the cases where the call extends past round $r_4$. Let $m'$ be the part of the message that remains to be sent. For $i \ge 1$ consider round $r_4 +i$. There are two possibilities:
\begin{itemize}
\item [(a)] After Alice has sent a key request in round $r_4 +i$, she hears silence and the call returns, or; 
\item[(b)] Alice correctly receives Bob's key, and uses it to send  the next piece of $m$ which Bob correctly receives.  Next, Alice hears silence from Bob.  If this was the last piece of the message, this ends the call. If not, the call continues into round $r_4+i+1$ with a shorter remaining message. 
\end{itemize}
Since the message has finite length, eventually (a) occurs. Thus, in every case there is some round $r_3 \geq r_1$ in which the call to \sendMessage(m) returns. \\

We now prove part (2) of our lemma.  We first show that the message $m$ is recorded at most once by Bob in the rounds $r_1$ to $r_3$.  By Lemma~\ref{lem:goodEvents} (2(c)), the bit $b'$ corresponding to the last partial message for $m$ received by Bob in line~\ref{alg:bobDecode-vl} of Algorithm~\ref{bobAlgo-vl} must be the same as the bit $b$ sent by Alice for this partial message from rounds $r_1$ to $r_3$.  Since this bit never changes, Bob will record $m$ (line~\ref{bob:record-vl} of Algorithm~\ref{bobAlgo-vl}) at most once in rounds $r_1$ to $r_3$. \\

If Bob \terminatedPi\ before round $r_3$, then part (2) of our lemma statement does not require Bob to record the message $m$, and so that part or our lemma is trivially true.  Thus, for the remainder of the proof, we assume that Bob has not \terminatedPi\ before round $r_3$ and we must show that he records the message $m$ exactly once.\\ 

We do this by induction on the length of $m$.  Note that since $m$ is a message from $\pi$, it belonged to $\pfl$, so it is actually possible for Bob to record $m$.  If $m$ is short enough to be sent in one piece, then  consider round $r_3$ in which the call to \sendMessage($m$) returns. Let $M_1$ be the string read by Alice on line~\ref{alg:AliceReadsBobsKey-vl} of Algorithm~\ref{aliceAlgo-vl}. Since Bob has not \terminatedPi, Lemma~\ref{lem:goodEvents} (2(b)) guarantees that Bob sent $M_1$, and therefore the call to Algorithm~\ref{aliceAlgo-vl} does not return on line~\ref{alg:AliceReturnsBecauseSilence}. Hence, it returns on line~\ref{alg:AliceReturnsFinally-vl}.  Since this is the last round, $M_1$ must correctly decode to $(k_B, k_A)$ in Algorithm~\ref{aliceAlgo-vl}. Thus, Alice sends Bob $m$ using $k_B$ and hears silence subsequently. Thus, Bob must have actually been silent at this time (Lemma~\ref{lem:goodEvents} (1)), which only happens if he has either now or previously recorded $m$. Hence, Bob must have recorded the message $m$ in some round $r_2 \leq r_3$. \\

Now as an induction hypothesis suppose the conclusion about Bob recording a message exactly once is true for all strings $s$ shorter than $m$. Suppose  $m$ requires more than one piece to be sent. Then since Alice continues to resend the first piece $m_0$ until she has received confirmation that Bob has received it, there is some round $r_1' < r_3$ in which Alice receives this confirmation. Since the same parity bit $b$ for this piece  has been sent in all rounds $r\le r_1'$ , and received bit $b'$ agrees with $b$ in each of these rounds, Bob stores the partial message $m_0$ exactly once. Let $m_1$ be the remaining portion of $m$. Clearly it is shorter than $m$. Now let us examine the control flow throughout the algorithm from round $r_1'$ onwards. This looks exactly like a call to \sendMessage($m_1$) with the persistent variable $b$ now set to $\lnot b$\jared{is this right?}, together with a \receiveMessage \ call in which the persistent variables $\hat{b}$ set to $\lnot \hat{b}$ \jared{correct??}, $\mu$ set to $m_0$, and $\lambda$ set to $|m_0|$. By induction hypothesis, there is exactly one round $r_2$ with $r_1' \le r_2 \le r_3$ during which Bob records $m_0 \circ m_1$, which is in $\pfl$. But since $m = m_0 \circ m_1$, this concludes the proof.
\end{proof}

\begin{lemma}
\label{lem:correctnessRounds-vl}
The following holds for protocol $\pi'$. For any message $m$ that Bob records (line~\ref{bob:record-vl} of Algorithm~\ref{bobAlgo-vl}) in some round $r_2$, Alice started a call to \sendMessage $(m)$ in a round $r_1 \leq r_2$ and returned from that call in some round $r_3 \geq r_2$. 	
\end{lemma}
\begin{proof}
Consider the round $r_2$ in which the Bob records a message $m$ (line~\ref{bob:record-vl} of Algorithm~\ref{bobAlgo-vl}). Suppose $\mu = \varnothing$ when \receiveMessage \ was called in round $r_2$. Then $m$ arose in the first component of $\decodeMsg{M_1'}{t}$ in line~\ref{alg:bobDecode-vl} of Algorithm~\ref{bobAlgo-vl}, where the second component correctly matched $k_B$. Specifically the first component must have been $(m, b) $ for some $b$. Thus, by Lemma~\ref{lem:goodEvents} (2(c)), Alice must have sent $\encodedMsg{(m,b)}{k_B}{t}$, during the transmission of $M_1'$.  Hence, Alice must be in a call to \sendMessage \ with some message $\hat{m}$ from $\pfl$, with  $m$ as a contiguous substring, and this call must have begun in some round $r_1 \leq r_2$. We must show that $\hat{m} =m$. Suppose $m$ is a proper contiguous substring of $\hat{m}$, \emph{i.e.}, $\hat{m}= m_0 \circ m \circ m_1$ where at least one of $m_0$ and $m_1$ is not $\varnothing$. Then $m_0$ is not in $\pfl$ (since it is a prefix of $\hat{m}$). Since Alice would not proceed to sending $m$ until she had confirmed that Bob had received $m_0$, she must have previously received that confirmation, meaning that Bob previously received $m_0$. But in that case, Bob would have that stored as a partial message, contradicting the assumption that $\mu = \varnothing$ when \receiveMessage \ was called in round $r_2$. Thus $m_0 = \varnothing$. But that means that $m$ is a prefix of $\hat{m}$. Since both $m$ and $\hat{m}$ are in $\pfl$,  it follows that $m_1 = \varnothing$ and $\hat{m} = m$. Thus, Alice called \sendMessage($m$). \\

Next suppose $\mu = s$ when \receiveMessage \ was called in round $r_2$, for some string $s \neq \varnothing$. Then $m = s\circ x$, where $x$ is the string Bob decodes in round $r_2$. We first induct on the length of $s$ to prove that Alice's call to \sendMessage \ had as input, some message $\hat{m}$ from $\pfl$, such that $\hat{m}$ contains  $s \circ x =  m$ as a contiguous substring, \emph{i.e.}, $\hat{m} = m_0 \circ m \circ m_1$ for some strings $m_0$, $m_1$. The case of $s = \varnothing$ was shown above. Thus, assume there exist rounds $\rho_0 < \rho_1 \leq r_2$ and non-empty strings $s_0$ and $s_1$, with $s = s_0 \circ s_1$ such that $\mu = s_0$ when \receiveMessage \ was called in round $\rho_0$ and $\mu = s_1$ when \receiveMessage \ was called in round $\rho_1$. Here, $\rho_0$ and $\rho_1$ are the first rounds with this property. Then Bob decoded $s_0$ in round $\rho_0-1$ and $s_1$ in round $\rho_1-1$. Furthermore, Bob has not decoded any other string between decoding $s_0$ and $s_1$. Then, it follows from the induction hypothesis that $s_0$ and $s_1$ came from Alice, so that Alice must be in call to \sendMessage \ during the course of which she sends encoded strings $s_0,s_1,x$ in this order. Moreover, she cannot have sent any other encoded strings in between, because had she done so, she would not have moved on from it without acknowledgement that Bob decoded it. Since we stipulate that Bob did not decode any other strings, it follows that the input to Alice's call to \sendMessage \ had $m = s_0 \circ s_1 \circ x$ as a contiguous substring so that $\hat{m} = m_0 \circ m \circ m_1$ for some strings $m_0$, $m_1$. Note that $m_0$ is not in $\pfl$. Now, when Bob decoded $s_0$, the call to \receiveMessage \ was passed $\varnothing$.  Thus, from the discussion in the beginning of the proof, it follows that $m_0 =\varnothing$ and therefore $m_1=\varnothing$ and $\hat{m}=m$. Thus, in all cases, Alice did actually have a call to \sendMessage($m$) that began in some round $r_1 \le r_2$.\\
  				
Finally, we know by Lemma~\ref{lem:msgDel-vl}(1) that every call to \sendMessage \ ends in some round $r_3$.  Since, by the above, Alice is in the call during round $r_2$, it must be the case that Alice returns from the call in some round $r_3 \geq r_2$.
\end{proof}

\subsection{Resource Costs}
We now analyze the total number of bits sent by $\pi'$ when the messages in $\pi$ may be longer than a single bit. Let $\alpha$ be the average message length in $\pi$. 

\begin{lemma}
	\label{lem:resourceCost2}
	If $\pi'$ succeeds, then it has the following resource costs.
	\begin{itemize}
	\item The number of bits sent is $\mathcal{O} \paran{L\paran{1 + \frac{1}{\alpha}\log \paran{\frac{n L}{\errorprob}} + T}}$.
	\item The latency is $\mathcal{O} \paran{\underset{p}{\max} \left \{ \Lambda_p\paran{1 + \frac{1}{\alpha}\log \paran{\frac{n(L+T)}{\errorprob}} + T_p \right \}}}$, where  $p$ is any communication path in the asynchronous simulation of $\pi$, $\Lambda_p$ is the latency of $p$, and $T_p$ is the total number of bits flipped by the adversary on edges in $p$.
	\end{itemize}

\end{lemma}
\begin{proof}

Call a string $M$ encoded in Step~\ref{alg:AliceKey-vl} of Algorithm~\ref{aliceAlgo-vl} a \emph{submessage}.  To upper-bound the number of bits sent, we assume pessimistically that in every round of $\pi'$, there is an attempt to send exactly one submessage.  This maximizes the number of bits sent since word sizes increase with time.  

Note that each message consists of at most $1$ submessage that is not of length some constant times the word length in that round.  Thus, the number of progressive rounds is no more than the number of messages sent, $L/\alpha$, plus the largest integer $x$ such that the number of bits sent in $x$ rounds equals $c_1 L$ for some constant $c_1 > 1$.  By Lemma~\ref{lem:numRoundsTil}, $x = O(L / \log \frac{nL}{\errorprob} )$.

The number of non-progressive rounds is $\mathcal{O}\paran{\frac{T}{\log \paran{n(T+1)/\errorprob}}}$, by the same argument as from the proof of Lemma~\ref{lem:resourceCosts1}.

Let the number of rounds $r = L/\alpha + O(L / \log \frac{nL}{\errorprob}) + \mathcal{O}\paran{\frac{T}{\log \paran{n(T+1)/\errorprob}}}$.  Then by Lemma~\ref{lem:taubound}, we can bound $\tau(r)$ as follows.

\begin{eqnarray*}
	\tau(r) & \leq &  \paran{L/\alpha + O\paran{\frac{L} {\log (nL / \errorprob)}} + \mathcal{O}\paran{\frac{T}{\log \paran{n(T+1)/\errorprob}}}}\paran{\log \frac{nr}{\errorprob}};\\
	& = & \mathcal{O} \paran{\paran{L/\alpha + \frac{L} {\log (nL / \errorprob)} + \frac{T}{\log \paran{n(T+1)/\errorprob}}}\paran{\log \frac{n(L+T)}{\errorprob}}};\\
	& = & \mathcal{O} \paran{\frac{1}{\alpha}\log \paran{\frac{n(L+T)}{\errorprob}} + L + T}.
\end{eqnarray*} 
 
Finally, to bound the latency, we note that the argument above holds for any communication path $p$ in the asynchronous simulation of $\pi$.  
 
\end{proof}

\begin{theorem}
	 The protocol $\pi'$ terminates with a correct simulation of $\pi$ with probability at least $1-\errorprob$.  If $\pi'$ terminates correctly, it has the following resource costs.
	\begin{itemize}
	\item The expected number of bits sent is $\mathcal{O} \paran{L\paran{1 + \frac{1}{\alpha}\log \paran{\frac{n(L+T)}{\errorprob}} + T}}$.
	\item The expected latency is $\mathcal{O} \paran{\underset{p}{\max} \left \{ \Lambda_p\paran{1 + \frac{1}{\alpha}\log \paran{\frac{n(L+T)}{\errorprob}} + T_p \right \}}}$, where  $p$ is any communication path in the asynchronous simulation of $\pi$, $\Lambda_p$ is the latency of $p$, and $T_p$ is the total number of bits flipped by the adversary on edges in $p$.
	\end{itemize}
\end{theorem}

\begin{proof}
	By Lemma~\ref{lem:failureTotal}, no failure event happens with probability at least $1-\errorprob$, and by Lemma~\ref{lem:correctness}, in such a situation, $\pi'$ terminates with an asynchronous simulation of $\pi$.  Upon correct termination, the resource cost bounds hold by Lemma~\ref{lem:resourceCost2}.
\end{proof}

\section{Conclusion and Future Work}
\label{sec:conclusion}
We have described the first algorithm in interactive communication for $n$ users that deals with the case of unknown number of bits sent by the protocol, while tolerating an unbounded but finite amount of noise. Against an adversary that flips $T$ bits, given an $\errorprob \in (0,1)$, our algorithm compiles a noise free protocol $\pi$ that sends $L$ bits into a robust protocol $\pi'$ that succeeds with probability $1-\errorprob$, and upon successful termination, sends $\mathcal{O} \paran{L\paran{1 + \frac{1}{\alpha}\log \paran{\frac{n(L+T)}{\errorprob}}} + T}$ bits, where $\alpha$ is the average message length in $\pi$. The blowup in the number of bits is constant for long messages in $\pi$ and within logarithmic factors of the optimal, otherwise. \\

\par Several open problems remain including the following. First, can we adapt our results to interactive communication where $\pi$ is a synchronous protocol? Second, can we handle an unknown amount of stochastic noise more efficiently, while making no assumption on the value of $L$ or $T$? Finally, for any algorithm, what is the minimum number of private random bits required to be hidden from the adversary to achieve robustness?

\printbibliography
\end{document}